\documentclass[12pt, draftclsnofoot, onecolumn]{IEEEtran}
\usepackage{amsfonts}
\usepackage{cite}
\usepackage[top=1 in, bottom=1 in, left=1 in, right= 1 in]{geometry}
\usepackage[cmex10]{amsmath}
\usepackage{amssymb}
\usepackage{array}
\usepackage{chemarrow}
\usepackage{times,epsfig}
\usepackage{graphicx}
\usepackage{subcaption}
\usepackage{setspace}
\usepackage{times,epsfig}
\usepackage{bbm}
\usepackage{verbatim}
\usepackage{array}
\usepackage{amsmath}
\usepackage{epstopdf}
\usepackage{mathtools}

\newtheorem{remark}{\textbf{Remark}}
\newtheorem{corollary}{\textbf{Corollary}}

\newtheorem{theorem}{\textbf{Theorem}}
\newtheorem{lemma}{\textbf{Lemma}}

\makeatletter

\begin{document}
	
	\title{Directional Cell Search Delay Analysis for Cellular Networks with Static Users}
	
	\author{ Yingzhe Li, Fran\c{c}ois Baccelli, Jeffrey G. Andrews, Jianzhong Charlie Zhang \thanks{Y. Li, J. G. Andrews and F. Baccelli are with the Wireless Networking and Communications Group (WNCG), The University of Texas at Austin (email: yzli@utexas.edu, francois.baccelli@austin.utexas.edu, jandrews@ece.utexas.edu). J. Zhang is with Samsung Research America-Dallas (email: jianzhong.z@samsung.com). Date revised: \today.}}
	
	\maketitle

\begin{abstract}
	 Cell search is the process for a user to detect its neighboring base stations (BSs) and make a cell selection decision. Due to the importance of beamforming gain in millimeter wave (mmWave) and massive MIMO cellular networks, the directional cell search delay performance is investigated. A cellular network with fixed BS and user locations is considered, so that strong temporal correlations exist for the SINR experienced at each BS and user. For Poisson cellular networks with Rayleigh fading channels, a closed-form expression for the spatially averaged mean cell search delay of all users is derived. This mean cell search delay for a noise-limited network (e.g., mmWave network) is proved to be infinite whenever the non-line-of-sight (NLOS) path loss exponent is larger than 2. For interference-limited networks, a phase transition for the mean cell search delay is shown to exist in terms of the number of BS antennas/beams $M$: the mean cell search delay is infinite when $M$ is smaller than a threshold and finite otherwise. Beam-sweeping is also demonstrated to be effective in decreasing the cell search delay, especially for the cell edge users.
\end{abstract}

\section{Introduction}
Cell search is a critical prerequisite to establish an initial connection between a cellular user and the cellular network. Specifically, the users will detect their neighboring BSs and make the cell selection decision during a downlink cell search phase, after which the users can acquire connections with the network by initiating an uplink random access phase. The transmissions and receptions during cell search are performed omni-directionally in LTE~\cite{dahlman2013LTEbook}, but this is unsuitable for mmWave communication~\cite{pi2011introduction,rappaport2013millimeter,roh2014millimeter,ghosh2014millimeter} or massive MIMO~\cite{Marzetta2010noncooperative,Larsson2014Massive,Rusek2013Scaling,Bjornson2016massive} due to the lack of enough directivity gain. By contrast, directional cell search schemes that leverage BS and/or user beam-sweeping to achieve extra directivity gains, can ensure reasonable cell search performance~\cite{andrews2014will,andrews2017modeling,barati2015directional,li2016design,Bjornson2016massive,Larsson2014operation,shepard2015control}. In this paper, we leverage stochastic geometry~\cite{baccelli2010stochasticpt2,trac,stochtutorial} to develop an analytical framework for the directional cell search delay performance of a \textit{fixed} cellular network, where the BS and user locations are fixed over a long period of time (e.g., more than several minutes). We believe the analytical tools developed in this paper can provide useful insights into practical fixed cellular networks such as fixed mmWave or massive MIMO broadband networks~\cite{Larsson2014Massive,Verizon20165G2,Pi2016Millimeter}, or mmWave backhauling networks~\cite{pi2011introduction,Hur2013millimeter}.


\subsection{Related Work}
Beam-sweeping is a useful method to improve cell search performance compared to conventional omni-directional cell search for both mmWave and massive MIMO networks. Specifically, mmWave links generally require high directionality with large antenna gains to overcome the high isotropic path loss of mmWave propagation. As a result, in mmWave networks, applying beam-sweeping for cell search not only provides sufficient signal-to-noise ratio (SNR) to create viable communications, but also facilitates beam alignment between the BS and users~\cite{andrews2017modeling,barati2015directional,giordani2016initial,Giordani2016comparative,li2016design,li2016performance}. The directional cell search delay performance of mmWave systems has been investigated in~\cite{barati2015directional,giordani2016initial,Giordani2016comparative} from a link level perspective, and in~\cite{li2016design,li2016performance} from a system level perspective. In particular,~\cite{li2016design} and \cite{li2016performance} consider the user and mmWave BS locations are fixed within an initial access cycle, but independently reshuffled across cycles. This block coherent scenario is fundamentally different from that of a fixed network. For a massive MIMO system, the BSs can achieve an effective power gain that scales with the number of antennas if the channel state information (CSI) is known at the BSs~\cite{Bjornson2016massive}. However, since such an array gain is unavailable for cell search operations due to the lack of CSI, the new users may be unable to join the system using the traditional omnidirectional cell search~\cite{Bjornson2016massive,Larsson2014operation,shepard2015control}. In order to overcome this issue,~\cite{shepard2015control} has proposed open-loop beamforming to exhaustively sweep through BS beams for cell search. This design has been implemented and verified on a sub-6 GHz massive MIMO prototype~\cite{shepard2015control}, but the analytical directional cell search performance has not been investigated for fixed cellular networks from a system level perspective. 

Due its analytical tractability for cellular networks~\cite{baccelli2010stochasticpt2,trac,stochtutorial}, stochastic geometry is a natural candidate for analyzing the directional cell search delay in such fixed cellular networks. In particular, stochastic geometry has already been widely used to investigate fixed Poisson network performance through the local delay metric~\cite{baccelli2010new,baccelli2010stochasticpt2,Haenggi2013local,Zhang2012delay,Iyer2015may}, which characterizes the number of time slots needed for the SINR to exceed a certain SINR level. In~\cite{baccelli2010new,baccelli2010stochasticpt2}, the local delay for fixed ad hoc networks was found to be infinite under several standard scenarios such as Rayleigh fading with constant noise. A new phase transition was identified for the interference-limited case in terms of the mean local delay: the latter is finite when certain parameters are above a threshold, and infinite otherwise. The local delay for noise-limited and interference-limited fixed Poisson networks was also investigated in~\cite{Haenggi2013local,Zhang2012delay,Iyer2015may}, where it is shown that power control is an efficient method to ensure a finite mean local delay. 
These previous works mainly focused on omni-directional communications.

\subsection{Contributions}
In this work, we analyze the cell search delay in \textit{fixed} cellular networks with a directional cell search protocol. We consider a time-division duplex (TDD) cellular system, where system time is divided into different initial access (IA) cycles. Each cycle starts with the cell search period, wherein BSs apply a synchronous beam-sweeping pattern to broadcast the synchronization signals. A mathematical framework is developed to derive the exact expression for the mean cell search delay, which quantifies the spatial average of the individual mean cell search delays perceived by all users.  
The main contributions of this paper are summarized as follows.

\textbf{Beam-sweeping is shown to reduce the number of IA cycles needed to succeed in cell search.} For any arbitrary BS locations and fading distribution, the mean number of initial access cycles required to succeed in cell search is proved to be decreasing when the number of BS antennas/beams is multiplied by a factor $m > 1$.

\textbf{An exact expression for the mean cell search delay is derived for Poisson point process (PPP) distributed BSs and Rayleigh fading channels.} This expression is given by an infinite series, based on which the following observations are obtained:
		\begin{itemize}
			\item Under the noise limited scenario (e.g., mmWave networks), we prove that as long as the path loss exponent for NLOS path is larger than 2, the mean cell search delay is infinite, irrespective of the BS transmit power and the BS antenna/beam number. 
			\item Under the interference limited scenario (e.g., massive MIMO networks in sub-6 GHz bands), there exists a phase transition for mean cell search delay in terms of the BS antenna/beam number $M$. Specifically, the mean cell search delay is infinite when $M$ is smaller than a critical value and finite otherwise. This fact was never observed in the literature to the best of our knowledge. 
		\end{itemize}
		
\textbf{Cell search delay distribution is numerically evaluated.} The conditional mean cell search delay of a typical user given its nearest BS distance is derived for PPP distributed BSs and Rayleigh fading channels. The distribution of this conditional mean cell search delay is also numerically evaluated, and we observe that the cell search delay distribution is heavy-tailed. We also show that increasing the number of BS antennas/beams can significantly reduce the cell search delay for cell edge users. 	

Overall, this paper has shown that in fixed networks the mean cell search delay could be very large due to the temporal correlations induced by common randomness. As a result, for fixed cellular networks, system parameters including the number of BS antennas and/or BS intensity need to be carefully designed for reasonable cell search delay performance to be achieved.

\section{System Model}\label{SysModelSec}
In this work, we consider a cellular system that has carrier frequency $f_c$ and total system bandwidth $W$. The BS transmit power is denoted by $P_b$, and the total thermal noise power is denoted by $\sigma^2$. In the rest of this section, we present the proposed directional cell search protocol, location models, propagation assumptions, and the performance metrics. 
\subsection{Directional Cell Search Protocol}\label{IAProtSubsec}
We consider a TDD cellular system as shown in Fig.~\ref{fig:time_Structure}, where system time is divided into different initial access cycles with period $T$, and where $\tau$ denotes the OFDM symbol period. Initial access refers to the procedures that establish an initial connection between a user and the cellular network. It consists of two main steps: cell search on the downlink and random access (RA) on the uplink. Specifically, by detecting the synchronization signals broadcasted by BSs during cell search, a user can determine the presence of its neighboring BSs and make the cell selection decision. Then the user can initiate the random access process to its desired serving BS by transmitting a RA preamble through the shared random access channel, and it is successfully connected to the network if the BS can decode the RA preamble without any collision. The main focus of this work is the cell search performance, while the random access performance will be incorporated in our future work. 


Each BS is equipped with a large dimensional antenna array with $M$ to support highly directional communications. For analytical tractability, the actual antenna pattern is approximated by a sectorized beam pattern, where the antenna gain is constant within the main lobe. In addition, we assume a $0$ side lobe gain for the BS, which is a reasonable approximation because the BS uses a large dimensional antenna array with narrow beams, possibly with a front-to-back ratio larger than 30 dB~\cite{waterhouse2002broadband}. 
Each BS supports analog beamforming with a maximum of $M$ possible BF vectors, where the $m$-th ($1 \leq m \leq M$) beamforming (BF) vector corresponds to the main-lobe, which has antenna gain $M$, and covers a sector area with angle $ [2\pi \frac{m-1 }{M}, 2\pi\frac{m}{M})$~\cite{alkhateeb2016initial}. Each user is assumed to have a single omni-directional antenna with unit antenna gain~\cite{shepard2015control,hussain2017throughput}. 


In the cell search phase, each BS sweeps through all $M$ transmit beamforming directions to broadcast the synchronization signals, and each user is able to detect a BS with sufficiently small miss detection probability (such as 1\%) if the signal-to-interference-plus-noise ratio (SINR) of the synchronization signal from that BS exceeds $\Gamma_{cs}$. 
All BSs transmit synchronously using the same beam direction during every symbol, and the cell search delay within each IA cycle is therefore $T_{cs} = M \times \tau$. When every BS transmits using the $m$-th ($1 \leq m \leq M$) BF direction, the typical user can only receive from the BSs located inside the ``BS sector"
\begin{align}
S\left(o,\frac{2\pi (m-1)}{M} + \pi, \frac{2\pi m}{M} + \pi\right),
\end{align} 
where we define the infinite sector domain centered at $u \in \mathbb{R}^2$ by:  
\begin{align}
S(u,\theta_1,\theta_2) = \{ x \in \mathbb{R}^2, \text{ s.t., } \angle(x - u) \in [\theta_1, \theta_2) \}.
\end{align}
There are $M$ such non-overlapping BS sectors during cell search, with the $j$-th ($ 1 \leq j \leq M$) sector being $S(o,\frac{2\pi (j-1)}{M}, \frac{2\pi j}{M})$.
We say a BS sector is detected during cell search if the typical user is able to detect the BS that provides the smallest path loss (i.e., the closest BS) inside this sector, where the path loss can be estimated from the beam reference signals~\cite{Verizon20165G2}. After cell search, the typical user selects the BS with the smallest path loss among all the detected BS sectors as its serving BS. For simplicity, we neglect the scenario where the BS providing the smallest path loss inside a BS sector is in deep fade and unable to be detected, while some other BSs can be detected in the same sector. Such a scenario does not change the fundamental trends regarding the finiteness of the mean cell search that will be detailed in Section~\ref{MeanCSDelaySec} (e.g. Theorem~\ref{NoiseLimLemma1}), but the corresponding analysis is significantly more complicated. 


\begin{figure}[h]
	\centering
	\includegraphics[width=0.8\linewidth]{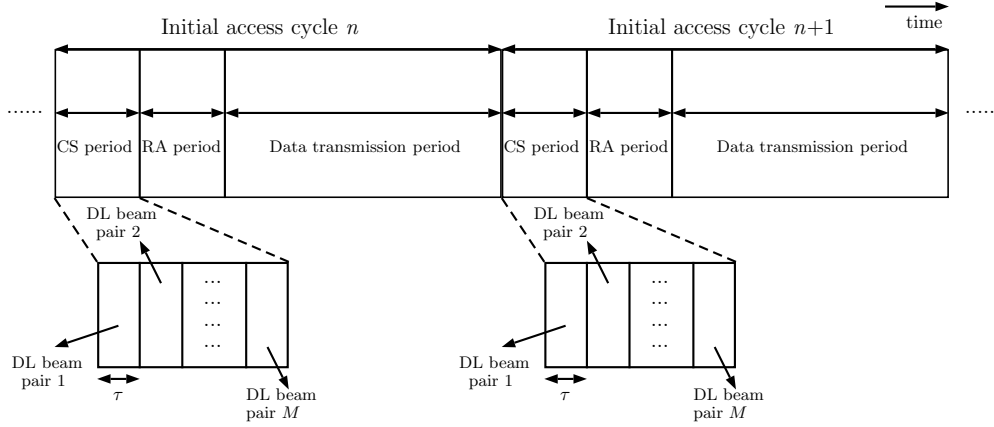}
	\caption{Illustration of two initial access cycles and the timing structure.}
	\label{fig:time_Structure}
\end{figure}


\begin{table}[t]
	\center\caption[Notation and Simulation Parameters]{Notation and Simulation Parameters}\label{SysParaTable}
	\resizebox{460pt}{!}{%
		\begin{tabular}{|c|p{130mm}|p{60mm}|}
			\hline 
			Symbol & Definition & Simulation Value \\ 
			\hline 
			$\Phi$, $\lambda$ & BS PPP and intensity & $\lambda = 100$ BS/km$^2$ \\ 
			\hline 
			$\Phi_u$, $\lambda_u$ & User PPP and intensity & $\lambda_u = 1000$ users/km$^2$ \\ 
			\hline
			$P_b$, $P_u$ & BS and user transmit power & $P_b = 30$ dBm, $P_u = 23$ dBm \\ 
			\hline
			$f_c, B$ & Carrier frequency and system bandwidth &  ($f_c, B$) = ($73$, $1$) GHz, ($2, 0.2$) GHz \\
			\hline 
			$W$ & Total thermal noise power& $-174 \text{   dBm/Hz} + 10 \log_{10}(B)$\\
			\hline
			$M$ & Number of BS antennas and BF directions supported at each BS & \\
			\hline 
			$\alpha_L, \alpha_N$ & Path loss exponents for dual-slope model & $(\alpha_L, \alpha_N) = (2.1, 3.3), (2.5,2.5)$\\ 
			\hline 
			$C_L, C_N$ & Path loss at close-in reference distance for dual-slope model & $(C_L, C_N) = (69.71, 69.71)$ dB, $(38.46, 38.46)$ dB \\ 
			\hline
			$R_c$ & Critical distance for dual-slope path loss model & 50m \\
			\hline 
			$\Gamma_{cs}, \Gamma_{ra}$ & SINR threshold to detect synchronization signal and RA preamble & $(\Gamma_{cs}, \Gamma_{ra}) = (-4,-4)$  dB \\ 
			\hline
			$\tau$ & OFDM symbol period & 14.3 $\mu$s, 71.4 $\mu$s  \\
			\hline
			$T$ & Initial access cycle period & 20 ms, 100 ms \\
			\hline
			$S_M(i)$ & $i$-th BS sector, i.e., $S_M(i) = \{ x \in \mathbb{R}^2, \text{ s.t., } \angle x \in [2\pi\frac{(i-1)}{M}, 2\pi \frac{i}{M}) \}$ & \\ 
			\hline
			$\{x_0^i\}_{i=1}^{M}$ & BS providing the smallest path loss to the typical user inside $S_M(i)$& \\
			\hline 
			$R_0$ & Distance from typical user to its nearest BS & \\
			\hline
			$L_{cs}(M,\lambda)$ & Number of IA cycles to succeed in cell search & \\
			\hline
			$L_{cs}(R_0.M,\lambda)$ & Mean number of IA cycles to succeed in cell search conditionally on $R_0$& \\
			\hline
			$D_{cs}(M,\lambda)$ & Cell search delay &\\
			\hline
			$B(x,r)$ ($B^{o}(x,r)$) & Closed (open) ball with center $x$ and radius $r$ & \\
			\hline
		\end{tabular} }
	\end{table}

\subsection{Spatial Locations and Propagation Models} The BS locations are assumed to be a realization of a stationary point process $\Phi = \{x_i\}_i$ with intensity $\lambda$. The user locations are modeled as a realization of a homogeneous PPP with intensity  $\lambda_u$, which is denoted by $\Phi_u = \{u_i\}_i$. 
In this paper, a fixed network scenario is investigated  where the BS locations are fixed, and the users are either fixed or move with very slow speed such as a pedestrian speed (e.g., less than $1$ km/h) . As a result, the BS and user locations appear to be fixed across different initial access cycles. This is fundamentally different from the high mobility scenario investigated in~\cite{li2016design,li2016performance}, which assumes the BS and user PPPs are independently re-shuffled across every initial access cycles. 

Without loss of generality, we can analyze the performance of a typical user $u_0$ located at the origin. This is guaranteed by Slivnyak's theorem, which states that the property observed by the typical point of a PPP $\Phi^{'}$ is the same as that observed by the point at origin in the process $\Phi^{'} \cup \{o\}$~\cite{chiu2013stochastic,baccelli2010stochastic}. 

A dual-slope, non-decreasing path loss function~\cite{zhang2015downlink} is adopted, where the path loss for a link with distance $r$ is given by: 
\allowdisplaybreaks
\begin{align}
\allowdisplaybreaks
&\mathit{l}(r)=
\left\{
\begin{array}{ll}
C_L r^{\alpha_L}, & \text{if  } r < R_C,\\
C_N r^{ \alpha_N}, & \text{if  } r \geq R_C.\\
\end{array}
\emph{ } \right.
\label{DualSlopePLEq}
\end{align}
The dual slope path loss model captures the dependency of the path loss exponent on the link distance for various network scenarios, such as ultra-dense~\cite{zhang2015downlink} and mmWave networks~\cite{bai2015coverage}. In particular,~(\ref{DualSlopePLEq}) is referred to as the LOS ball blockage model for mmWave networks~\cite{zhang2015downlink}, wherein $\alpha_L$ and $\alpha_N$ represent the LOS and NLOS path loss exponents, and $C_L$ and $C_N$ represent the path loss at a close-in reference distance (e.g., 1 meter). We focus on the scenario where $\alpha_N \geq \max(\alpha_L,2)$. If $\alpha_L = \alpha_N = \alpha$ and $C_L = C_N = C$, the dual slope path loss model reverts to the standard single-slope path loss model. 

Due to the adopted antenna pattern for BSs, the directivity gain between BS and user is $M$ when the BS beam is aligned with the user, and $0$ otherwise. The fading effect for every BS-user link is modeled by an i.i.d. random variable, whose complementary cumulative distribution function (CCDF) is a decreasing function $G(\cdot)$ with support $[0,\infty)$. In addition, we assume the IA cycle length is such that the fading random variables for a given link are also i.i.d. across different cycles. 


\subsection{Performance Metrics}
The main performance metrics investigated in this work are the number of IA cycles, and the corresponding cell search delay for the typical user to discover its neighboring BSs and determine a potential serving BS. 
Without loss of generality, the IA cycle $1$ in Fig.~\ref{fig:time_Structure} represents the first IA cycle of the typical user. Denote by $e_M(n)$ 
the success indicator for cell search of IA cycle $n$. The number of IA cycles for the typical user to succeed in cell search is therefore:
\begin{align}
L_{cs}(M,\lambda) = \inf \{n \geq 1: e_M(n) = 1\}.
\end{align}Since analog beamforming is adopted at each BS, the cell search delay is defined as follows:
\begin{align}\label{CSandIADelayDefnEq}
D_{cs}(M,\lambda) &= (L_{cs}(M,\lambda) -1)  T + M \tau.
\end{align}

Finally, Table~\ref{SysParaTable} summarizes the notation, the definitions and the system parameters that will be used in the rest of this paper\footnote{For the symbols with two simulation values, the first one is for the noise limited scenario, and the second one is for the interference limited scenario, which will be detailed in Section~\ref{NumEval}.}.

\section{Analysis for Mean Cell Search Delay}\label{MeanCSDelaySec}
In this section, the mean cell search delay performance for the typical user is investigated, which corresponds to the cell search delay under the Palm expectation with respect to the user PPP $\Phi_u $ (i.e., $\mathbb{E}_{\Phi_u}^0[D_{cs}(M,\lambda)]$). In fact, the Palm expectation can also be understood from its ergodic interpretation, which states that for any user $u \in \Phi_u$ with cell search delay $D_{cs}(u,M,\lambda)$, the following relation is true: 
\begin{align}
\mathbb{E}_{\Phi_u}^0[D_{cs}(M,\lambda)] = \lim_{n\rightarrow \infty} \frac{1}{\Phi_u(B(0,n))} \sum_{k} \textbf{1}_{u \in B(0,n)} D_{cs}(u,M,\lambda).
\end{align}
Therefore, the mean cell search delay of the typical user can also be understood as the spatial average of the individual cell search delays among all the users. 
For notational simplicity, we will use $\mathbb{E}$ in the rest of this paper to denote the Palm expectation under the user PPP $\Phi_u$.

\subsection{Cell Search Delay Under General BS Deployment and Fading Assumptions}
In this part, we first investigate the cell search delay under a general BS location model (not necessarily PPP) and fading distribution. According to Section~\ref{SysModelSec}, the BS and user locations are fixed, and the fading variables for every link are i.i.d. across IA cycles. Therefore, given the BS process $\Phi$, the cell search success indicators for different IA cycles $\{e_M(n)\}$ form an i.i.d. Bernoulli sequence of random variables. The cell search success probability is denoted by $\pi_{M}(\Phi) = \mathbb{E}\left[e_M(1) | \Phi\right]$.

Since each BS sector can be independently detected given $\Phi$, and cell search is successful if at least one BS sector is detected. Conditionally on $\Phi$, the cell search success probability for every IA cycle is therefore: 
\begin{align}
\pi_{M}(\Phi) = 1 - \prod_{i = 1}^{M} \left[ 1- \mathbb{E}\left[\hat{e}_{M}(i) | \Phi\right]\right],
\end{align}
where $\hat{e}_{M}(i)$ denotes the indicator that the BS providing the smallest path loss inside BS sector $i$ is detected. Specifically, if we denote by $S_M(i) \triangleq S(o,\frac{2\pi (i-1)}{M}, \frac{2\pi i}{M})$ the BS sector $i$, $x_0^{i}$ the BS providing the smallest path loss to the typical user in $\Phi \cap S_M(i)$, and by $\{F_j^i\}$ the fading random variables from BSs in $S_M(i)$ to the typical user, we have:
\begin{align}
\mathbb{E}\left[\hat{e}_{M}(i) | \Phi\right] &= \mathbb{P}\biggl(\frac{F_{0}^{i}/\mathit{l}(\|x_0^i\|)}{\sum\limits_{x_j^{i} \in \Phi \cap S_M(i) \setminus \{x_0^i\}}\mathit{F}_{j}^{i}/\mathit{l}(\|x_j^{i}\|) + W/PM} > \Gamma_{cs} \biggl| \Phi\biggl)\nonumber \\
&=\mathbb{E}\biggl[G\biggl(\Gamma_{cs} \mathit{l}(\|x_0^i\|) (\sum\limits_{x_j^{i} \in \Phi \cap S_M(i) \setminus \{x_0^i\}}\mathit{F}_{j}^{i}/\mathit{l}(\|x_j^{i}\|) + W/PM)\biggl)\biggl|\Phi\biggl]\label{SPGeneralEq2},
\end{align}
where the expectation in~(\ref{SPGeneralEq2}) is taken with respect to the i.i.d. fading random variables $\{F_j^i\}$. 

In the following theorem, we derive the mean number of IA cycles for the typical user to succeed in the cell search under the Palm expectation of the user process. 
\begin{theorem}\label{IADelayThm1}
	The mean number of IA cycles needed for the typical user to succeed in cell search is given by:
	\begin{align}
	\allowdisplaybreaks
	&\mathbb{E}[L_{cs}(M,\lambda) | \Phi] = \frac{1}{1 - \prod_{i = 1}^{M} \left[ 1- \mathbb{E}\left[\hat{e}_{M}(i) | \Phi\right]\right]},\label{IADelayCondEq1}\\
	&\mathbb{E}[L_{cs}(M,\lambda)] = \mathbb{E}_{\Phi}\left[\frac{1}{1 - \prod_{i = 1}^{M} \left[ 1- \mathbb{E}\left[\hat{e}_{M}(i) | \Phi\right]\right]}\right]\label{IADelayCondEq2}.
	\end{align}
\end{theorem}
\begin{proof}
	The first part can be proved by the fact that given $\Phi$, $L_{cs}(M,\lambda)$ has a geometric distribution with success probability  $\pi_{M}(\Phi)$; while the second part follows by taking the expectation of~(\ref{IADelayCondEq1}) with respect to $\Phi$.
\end{proof}

\begin{remark}
	Since $\mathbb{E}\left[\hat{e}_{M}(i)|\Phi\right] > 0$ according to~(\ref{SPGeneralEq2}), the conditional mean cell search delay $\mathbb{E}[L_{cs}(M,\lambda) | \Phi]$ will be finite almost surely. However, the overall spatial averaged mean cell search delay with respect to (w.r.t.) the BS PPP $\Phi$ (i.e. $\mathbb{E}[L_{cs}(M,\lambda)]$) could be infinite under certain network settings. This will be detailed in the next subsection.
\end{remark}

A lower bound and an upper bound to $\mathbb{E}[L_{cs}(M,\lambda)]$ can be immediately obtained from~(\ref{IADelayCondEq2}), which are provided in the following remarks. 
\begin{remark}\label{UnivLBRmk}
	By applying Jensen's inequality to the positive random variable $X$ and the function $f(x) = \frac{1}{x}$, we get that $\mathbb{E}[\frac{1}{X}] \geq \frac{1}{\mathbb{E}[X]}$. Thus
	\begin{align}\label{GeneralUB}
	\mathbb{E}[L_{cs}(M,\lambda)] \geq \frac{1}{1 - \mathbb{E}\left[\prod_{i = 1}^{M} \left[ 1- \mathbb{E}\left[\hat{e}_{M}(i)|\Phi\right]\right]\right]},
	\end{align}
	where the equality holds when the BS PPP is independently re-shuffled across different IA cycles from the typical user's perspective, which coincides with the high mobility scenario considered in~\cite{li2016design,li2016performance}. 
\end{remark}
	

\begin{remark}\label{UnivUBRmk}
	If we denote by $x_0$ the BS providing the smallest path loss to the typical user, and $i^{*}$ the index for the BS sector that contains $x_0$, then $\prod_{j = 1}^{M} \left[ 1- \mathbb{E}\left[\hat{e}_{M}(j) | \Phi\right]\right] \leq 1- \mathbb{E}\left[\hat{e}_{M}(i^*) | \Phi\right]$. Therefore, an upper bound to $\mathbb{E}[L_{cs}(M,\lambda)]$ is given by:
	\begin{align}\label{UnivUBEq}
		\mathbb{E}[L_{cs}(M,\lambda)] \leq \mathbb{E}\biggl[\frac{1}{\mathbb{E}\left[\hat{e}_{M}(i^*)|\Phi\right]}\biggl].
	\end{align}
\end{remark}

Based on Theorem~\ref{IADelayThm1}, we can prove the following relation between the BS antenna/beam number $M$ and the mean cell search delay. 
\begin{lemma}\label{BeamSweepGeneralLemma}
	Given a realization of BS locations $\Phi$, the mean number of IA cycles to succeed in cell search is such that $\mathbb{E}[L_{cs}(M_2,\lambda) | \Phi] < \mathbb{E}[L_{cs}(M_1,\lambda) | \Phi]$, if $M_2 = m M_1$ with $m$ being an integer larger than 1.  
\end{lemma}
\begin{proof}
	Since $M_2 = m M_1$, we know that $S_{M_1}(i) = \bigcup_{j=1}^{m}S_{M_2}((i-1)m+j)$ for $1 \leq i \leq M_1$. Denote by $x_0^i$ the BS providing the smallest path loss to the typical user inside $\Phi \cap S_{M_1}(i)$, and assume $x_0^i \in \Phi \cap S_{M_2}((i-1)m+j_0)$ for some $j_0 \in [1,m]$. Due to the facts that $M_2 > M_1$, $S_{M_2}((i-1)m+j_0) \subsetneq S_{M_1}(i)$, and since $G(\cdot)$ is a decreasing function, we get from~(\ref{SPGeneralEq2}) that $\mathbb{E}\left[\hat{e}_{M_1}(i) | \Phi\right] < \mathbb{E}\left[\hat{e}_{M_2}((i-1)m+j_0) | \Phi\right]$. Also note that $\mathbb{E}\left[\hat{e}_{M_2}((i-1)m+j) | \Phi\right] > 0$ for $\forall j \neq j_0$ according to~(\ref{SPGeneralEq2}), we hence have: 
	\begin{align}
	\allowdisplaybreaks
	\prod_{j = 1}^{m} \left[1- \mathbb{E}\left[\hat{e}_{M_2}\left((i-1)m + j\right) | \Phi\right]\right] < 1- \mathbb{E}\left[\hat{e}_{M_2}(i) | \Phi\right].
	\end{align}
	Thus the cell search success probability for the typical IA cycle satisfies: 
	\allowdisplaybreaks
	\begin{align}
	\allowdisplaybreaks
	\pi_{M_2}(\Phi) &=  1 - \prod_{i = 1}^{M_2} \left[ 1- \mathbb{E}\left[\hat{e}_{M_2}(i) | \Phi\right]\right] \nonumber \\
	 &= 1 - \prod_{i = 1}^{M_1} \left( \prod_{j = 1}^{m} \left[1- \mathbb{E}\biggl[\hat{e}_{M_2}\left((i-1)m + j\right) \biggl| \Phi\biggl]\right]\right) \nonumber \\
	 &> 1 - \prod_{i = 1}^{M_1}  \left[1- \mathbb{E}\left[\hat{e}_{M_1}(i) | \Phi\right]\right] = \pi_{M_1}(\Phi).\label{SPLemmaEq1}
	\end{align}
	Finally the proof is concluded by applying Theorem~\ref{IADelayThm1}.
\end{proof}

Lemma~\ref{BeamSweepGeneralLemma} shows that for all BS location models and fading distributions, the conditional number of IA cycles for cell search to succeed decreases when the number of BS antenna/beams is multiplied by an integer $m > 1$, or equivalently when the BS beamwidth is divided by $m$. This result also implies that $\mathbb{E}[L_{cs}(M_2,\lambda)] \leq \mathbb{E}[L_{cs}(M_1,\lambda)]$ if $M_2 = m M_1$. 

\begin{remark}
	In fact, Lemma~\ref{BeamSweepGeneralLemma} cannot be further extended. If $M_2 > M_1$ but $M_2/M_1$ is not an integer, there will always exist special constructions of BS deployments such that $\mathbb{E}[L_{cs}(M_2,\lambda) | \Phi] > \mathbb{E}[L_{cs}(M_1,\lambda) | \Phi]$. 
\end{remark}
%

For the rest of this section, we investigate the mean cell search delay under several specific network scenarios. 
\subsection{Mean Cell Search Delay in Poisson Networks with Rayleigh Fading}
In this part, the BS locations are assumed to form a homogeneous PPP with intensity $\lambda$, and the fading random variables are exponentially distributed with unit mean (i.e., $G(x) = \exp(-x)$). Due to its high analytical tractability, this network setting has been widely adopted to obtain the fundamental design insights for conventional macro cellular networks~\cite{trac}, ultra-dense cellular networks~\cite{zhang2015downlink}, and even mmWave cellular networks\footnote{The SINR and rate trends for mmWave networks under Rayleigh fading and PPP configured BSs have been shown to be close to more realistic fading assumptions, such as the Nakagami fading or log-normal shadowing~\cite{andrews2017modeling}.}~\cite{bai2015coverage,DiRenzo2015Stochastic}. 

Due to the PPP assumption for BSs, and the fact that different BS sectors are non-overlapping, every BS sector can therefore be  detected independently with the same probability. Since the path loss function $\mathit{l}(r)$ is non-decreasing, the BS that provides the minimum path loss to the typical user inside the $i$-th BS sector $\Phi \cap S_M(i)$ (i.e., $x_0^i$) is the closest BS to the origin. The angle of $x_0^i$ is uniformly distributed within $[2\pi(i-1)/M, 2\pi i/M)$, and the CCDF for the norm of $x_0^i$ can be derived as follows:
\begin{align}\label{PLDistCDFEq}
\allowdisplaybreaks
\mathbb{P}(\|x_0^i\| \geq r) &= \mathbb{P}\left( \min_{x \in \Phi \cap S_M(i)} \|x\| \geq r \right)=\exp(-\frac{\lambda \pi r^2}{M}),
\end{align}
where the second equality follows from the void probability for PPPs. Therefore, the probability distribution function (PDF) for $\|x_0^i\|$ is given by: 
\begin{align}\label{PLDistPDFEq}
f_{\|x_0^i\|}(r) = \frac{2\lambda \pi r}{M} \exp(-\frac{\lambda \pi r^2}{M}).
\end{align}

By applying $\Phi \sim \text{PPP}(\lambda)$ and $G(x) = \exp(-x)$ into~(\ref{SPGeneralEq2}), the conditional detection probability for the $i$-th BS sector is given by:
\allowdisplaybreaks
\begin{align}\label{CSSuccProbEq1}
\allowdisplaybreaks
\!\!\!\!\!\mathbb{E}\left[\hat{e}_{M}(i) | \Phi\right] &= \mathbb{E}\biggl[\exp\biggl( -\Gamma_{cs} \mathit{l}(\|x_0^i\|) \biggl(\sum\limits_{x_j^{i} \in \Phi \cap S_M(i) \setminus \{x_0^i\}}\mathit{F}_{j}^{i}/\mathit{l}(\|x_j^{i}\|) + W/PM \biggl)\biggl) \biggl| \Phi\biggl] \nonumber\\ 
&=\exp\biggl(-\frac{W\Gamma_{cs}\mathit{l}(\|x_0^i\|)}{PM}\biggl)\mathbb{E}\biggl[\prod_{x_j^{i} \in \Phi \cap S_M(i) \setminus \{x_0^i\}} \exp\biggl(-\Gamma_{cs} \mathit{l}(\|x_0^i\|)\mathit{F}_{j}^{i}/\mathit{l}(\|x_j^i\|)\biggl) \biggl] \nonumber \\
&\overset{(a)}{=}\exp\biggl(-\frac{W\Gamma_{cs}\mathit{l}(\|x_0^i\|)}{PM}\biggl) \prod_{x_j^{i} \in \Phi \cap S_M(i) \setminus \{x_0^i\}} \frac{1}{1+\Gamma_{cs} \mathit{l}(\|x_0^i\|)/\mathit{l}(\|x_j^i\|)} \triangleq F_M(i,\Phi),
\end{align}
where step (a) is obtained by taking the expectation w.r.t. the fading random variables.
\begin{theorem}\label{ExpCSProbThm}
	If $\Phi \sim \text{PPP}(\lambda)$, and the fading variables are exponentially distributed with unit mean, the mean number of cycles for cell search to succeed is:
	\begin{align}\label{ExpCSProbThmEq1}
	\mathbb{E}[L_{cs}(M,\lambda)] = \sum\limits_{j=0}^{\infty} A_j^M,
	\end{align}
	where $A_j = \mathbb{E}[(1-F_M(1,\Phi))^j]$ is given by:
	\begin{align}\label{ExpCSProbThmEq2}
	\allowdisplaybreaks
	A_j = &\int_{0}^{\infty}  \biggl\{\sum\limits_{k=0}^{j} (-1)^k {j \choose k} \exp\biggl(-\frac{Wk\Gamma_{cs} \mathit{l}(r_1) }{PM}\biggl) \exp\biggl(-\frac{2\pi\lambda}{M} \int_{r_1}^{\infty} \bigg(1-\frac{1}{(1+\Gamma_{cs} \mathit{l}(r_1)/\mathit{l}(r))^k}\bigg)r {\rm d}r\biggr)\biggl\}\nonumber \\
	& \times \frac{2\lambda \pi r_1}{M} \exp(-\frac{\lambda \pi r_1^2}{M}) {\rm d}r_1.
	\end{align}
\end{theorem}

\begin{proof}
By substituting~(\ref{CSSuccProbEq1}) into Theorem~\ref{IADelayThm1}, we obtain:	
\begin{align}
\allowdisplaybreaks
\mathbb{E}[L_{cs}(M,\lambda)] &= \mathbb{E}\biggl[\frac{1}{1 - \prod_{i = 1}^{M} \left[ 1- F_M(i,\Phi)\right]}\biggl]\nonumber \\
&\overset{(a)}{=}\mathbb{E}\biggl[\sum_{j=0}^{\infty}  \biggl(\prod_{i = 1}^{M} \left[ 1- F_M(i,\Phi)\right] \biggl)^j\biggl] \nonumber \\
&\overset{(b)}{=} \sum_{j=0}^{\infty} \mathbb{E}\biggl[\biggl(\prod_{i = 1}^{M} \left[ 1- F_M(i,\Phi)\right] \biggl)^j\biggl] \nonumber \\
&\overset{(c)}{=} \sum_{j=0}^{\infty} \biggl\{\mathbb{E}\biggl[ \biggl( 1- F_M(1,\Phi)\biggl)^j\biggl]\biggl\}^M,
\end{align}
where step (a) is derived from the fact that $\frac{1}{1-x} = \sum_{j=0}^{\infty} x^j$ for $0 \leq x < 1$, step (b) follows from the monotone convergence theorem, and step (c) is because the events for BS sectors to be detected are i.i.d. for PPP distributed BSs. Furthermore, we can compute $A_j$ as follows:
\allowdisplaybreaks
\begin{align}\label{CSProbProofEq2}
\allowdisplaybreaks
&\mathbb{E}\biggl[ \biggl( 1- F_M(1,\Phi)\biggl)^j\biggl] \nonumber \\
=&\int_{0}^{\infty} \mathbb{E}\biggl[ \biggl( 1- F_M(1,\Phi)\biggl)^j \biggl| x_0^1 = (r_1,0) \biggl] \frac{2\lambda \pi r_1}{M} \exp(-\frac{\lambda \pi r_1^2}{M}) {\rm d}r \nonumber\\
\overset{(a)}{=}&\int_{0}^{\infty} \sum_{k=0}^{j} (-1)^k {j \choose k}\mathbb{E}_{\Phi}^{x_0^1}\left[ \left(F_M(1,\Phi)\right)^k \biggl| \Phi \cap S_M(1) \cap B(o,r_1) = 0\right] \frac{2\lambda \pi r_1}{M} \exp(-\frac{\lambda \pi r_1^2}{M}) {\rm d}r \nonumber \\
\overset{(b)}{=}& \int_{0}^{\infty} \sum_{k=0}^{j} (-1)^k {j \choose k} \mathbb{E} \biggl[\exp\left(-\frac{Wk\Gamma_{cs}\mathit{l}(r_1)}{PM}\right) \!\!\!\!\! \prod_{x_j^{i} \in \Phi \cap S_M(i) \cap B^c(o,r_1)}  \!\!\!\!\!\frac{1}{(1+\Gamma_{cs} \mathit{l}(r_1)/\mathit{l}(\|x_j^i\|))^k}\biggl]\nonumber\\
&\times \frac{2\lambda \pi r_1}{M} \exp(-\frac{\lambda \pi r_1^2}{M}) {\rm d}r,
\end{align}
where $\mathbb{E}_{\Phi}^{x_0^1}[\cdot]$ in (a) denotes the expectation under the Palm distribution at BS $x_0^1$; and step (b) is derived from Slivnyak's theorem. Finally the proof can be concluded by applying the probability generating functional (PGFL) of PPPs~\cite{chiu2013stochastic} to~(\ref{CSProbProofEq2}).
\end{proof}
\begin{remark}
	Theorem~\ref{ExpCSProbThm} can be interpreted as $\mathbb{E}[L_{cs}(M,\lambda)] = \sum_{j=0}^{\infty}\mathbb{P}(L_{cs}(M,\lambda) > j)$, with $A_j^M$ in~(\ref{ExpCSProbThmEq1}) representing the probability that the BS sectors are not detected within $j$ IA cycles, i.e., $\mathbb{P}(L_{cs}(M,\lambda) > j)$. 
\end{remark}

Theorem~\ref{ExpCSProbThm} provides a series representation of the expected number of IA cycles to succeed cell search. However, 
it is unclear from Theorem~\ref{ExpCSProbThm} whether $\mathbb{E}[L_{cs}(M)]$ is finite or not. 
In the following, we will investigate the finiteness of $\mathbb{E}[L_{cs}(M,\lambda)]$ under two representative network scenarios, namely the noise limited scenario and the interference limited scenario. 
\subsubsection{Noise limited Scenario}
In the noise limited scenario, we assume the noise power dominates the interference power (or interference power is perfectly canceled), such that only noise power needs to be taken into account. Compared to conventional micro-wave cellular networks that operate in sub-6 GHz bands, mmWave networks have much higher noise power due to the wider bandwidth, and the interference power is much smaller due to the high isotropic path loss in mmWave. As a result, mmWave cellular networks are typically noise limited, especially when the carrier frequency and system bandwidth are high enough (e.g. 73 GHz carrier frequency with 2 GHz bandwidth)~\cite{singh2015tractable,bai2015coverage}. 

Since the interference power is zero under the noise limited scenario, Theorem~\ref{ExpCSProbThm} becomes:
\allowdisplaybreaks
\begin{align}\label{NumSlotCSNoiseLimitedEq1}
\allowdisplaybreaks
\mathbb{E}[L_{cs}(M,\lambda)] = \sum\limits_{j=0}^{\infty} \biggl\{\int_{0}^{\infty}  \biggl(1-\exp\left(-\frac{W\Gamma_{cs} \mathit{l}(r_1) }{PM}\right) \biggl)^{j}\frac{2\lambda \pi r_1}{M} \exp(-\frac{\lambda \pi r_1^2}{M}) {\rm d}r_1 \biggl\}^M.
\end{align}
Through the change of variable ($v = \lambda r^2$),~(\ref{NumSlotCSNoiseLimitedEq1}) becomes 
\begin{align}\label{NumSlotCSNoiseLimitedEq2}
\allowdisplaybreaks
\mathbb{E}[L_{cs}(M,\lambda)] = \sum\limits_{j=0}^{\infty} \biggl\{\int_{0}^{\infty}  \biggl(1-\exp\left(-\frac{W\Gamma_{cs} \mathit{l}(\sqrt{(v/\lambda)}) }{PM}\right) \biggl)^{j}\frac{2 \pi}{M} \exp(-\frac{ \pi v}{M}) {\rm d}v\biggl\}^M,
\end{align}
which shows that $\mathbb{E}[L_{cs}(M,\lambda)]$ is non-increasing as the BS intensity $\lambda$ increases, i.e., network densification helps in reducing the number of IA cycles to succeed in cell search. 

In the next two lemmas, we prove that the finiteness of $\mathbb{E}[L_{cs}(M,\lambda)]$ depends on the NLOS path loss exponent $\alpha_N$, and that a phase transition for $\mathbb{E}[L_{cs}(M,\lambda)]$ happens when $\alpha_N =2$.
\begin{theorem}\label{NoiseLimLemma1}
	Under the noise limited scenario, for any finite number of BS antennas/beams $M$ and BS intensity $\lambda$, $\mathbb{E}[L_{cs}(M,\lambda)] = \infty$ whenever the NLOS path loss exponent $\alpha_N > 2$. 
\end{theorem}
\begin{proof}
Given the number of BS antennas/beams $M$ and for any arbitrarily large positive value $v_0$ with $v_0 > R_c$, we can re-write~(\ref{NumSlotCSNoiseLimitedEq1}) to obtain the following lower bound on $\mathbb{E}[L_{cs}(M)]$:
\allowdisplaybreaks
\begin{align}\label{NoiseLimitedLemmaEq2}
\allowdisplaybreaks
&\sum\limits_{j=0}^{\infty} \biggl\{\int_{0}^{\infty}  \biggl(1-\exp\left(-\frac{W\Gamma_{cs} \mathit{l}(r_1) }{PM}\right) \biggl)^{j}\frac{2\lambda \pi r_1}{M} \exp(-\frac{\lambda \pi r_1^2}{M}) {\rm d}r_1 \biggl\}^M \nonumber\\
\overset{(a)}{\geq}& \sum\limits_{j=0}^{\infty} \biggl\{\int_{v_0}^{\infty}  \biggl(1-\exp\left(-\frac{W\Gamma_{cs} C_N r_1^{\alpha_N} }{PM}\right) \biggl)^{j}\frac{2\lambda \pi r_1}{M} \exp(-\frac{\lambda \pi r_1^2}{M}) {\rm d}r_1 \biggl\}^M  \nonumber\\
> & \sum\limits_{j=0}^{\infty}  \biggl\{ \biggl(1-\exp\left(-\frac{W\Gamma_{cs} C_N v_0^{\alpha_N} }{PM}\right) \biggl)^{j} \int_{v_0}^{\infty}  \frac{2\lambda \pi r_1}{M} \exp(-\frac{\lambda \pi r_1^2}{M}) {\rm d}r_1 \biggl\}^M  \nonumber\\
= & \sum\limits_{j=0}^{\infty} \biggl(1-\exp\left(-\frac{W\Gamma_{cs} C_N v_0^{\alpha_N} }{PM}\right) \biggl)^{jM}  \exp(-\lambda \pi v_0^2)  \nonumber\\
= & \frac{ \exp(-\lambda \pi v_0^2) }{1-(1-\exp(-W\Gamma_{cs} C_N v_0^{\alpha_N}/PM))^M} \nonumber\\
\overset{(b)}{\geq}& \frac{1}{M} \exp\left(W\Gamma_{cs} C_N v_0^{\alpha_N}/PM -  \lambda \pi v_0^2\right) \overset{v_0\rightarrow\infty}{\longrightarrow} \infty,
\end{align}
where $\mathit{l}(r_1) = C_N r_{1}^{\alpha_N}$ in step (a) because $r_1 \geq v_0 > R_c$. Step (b) follows from the fact that for any $0 \leq x \leq 1$ and $M \in \mathbb{N}^{+}$, we have: $(1-x)^M + xM \geq 1$, thus $\frac{1}{1-(1-x)^M} \geq \frac{1}{xM}$. Note that since $\alpha_N >2$,~(\ref{NoiseLimitedLemmaEq2}) goes to infinity when $v_0$ goes to infinity, which completes the proof. 
\end{proof}

According to Lemma~\ref{NoiseLimLemma1}, the expected cell search delay is infinity whenever $\alpha_N > 2$, which cannot be alleviated by BS densification (i.e., increase $\lambda$), or using a higher number of BS antennas (i.e., increase $M$). The reason can be explained from~(\ref{NoiseLimitedLemmaEq2}), which shows that due to the PPP-configured BS deployment, the typical user could be located at the ``cell edge" with its closest BS inside every BS sector farther than some arbitrarily large distance $v$. There is a $\exp(-\lambda \pi v^2)$ fraction of such cell edge users, and the corresponding number of IA cycles required for them to succeed in cell search is at least $\exp(C v^{\alpha_N})$ for some $C>0$. Therefore, the expected cell search delay averaged over all the users will ultimately be infinite when $\alpha_N>2$. From a system level perspective, this indicates that for noise limited networks with $\alpha_N > 2$, there will always be a significant fraction of cell edge users requiring a very large number of IA cycle to succeed cell search, so that the spatial averaged cell search delay perceived by all users will be determined largely by these cell edge users, which explains why an infinite mean cell search delay is observed.

\begin{theorem}~\label{NoiseLmdPLE2Lem}
	Under the noise limited scenario with NLOS path loss exponent $\alpha_N = 2$, the expected number of IA cycles to succeed in cell search $\mathbb{E}[L_{cs}(M,\lambda)] = \infty$ if the BS density $\lambda$ and the BS antenna/beam number $M$ satisfy $\lambda M < \frac{\Gamma_{cs}C_N W}{P\pi}$, and $\mathbb{E}[L_{cs}(M,\lambda)] < \infty$ if $\lambda M > \frac{\Gamma_{cs}C_N W}{P\pi}$, i.e., the phase transition for $\mathbb{E}[L_{cs}(M,\lambda)]$ happens at $(\lambda^*, M^*)$ with $\lambda^* M^*= \frac{\Gamma_{cs}C_N W}{P\pi}$.
\end{theorem}
\begin{proof}
	If $\alpha_N = 2$, it is clear from~(\ref{NoiseLimitedLemmaEq2}) that $\mathbb{E}[L_{cs}(M,\lambda)] = \infty$ if $\lambda M< \frac{\Gamma_{cs}C_N W}{P\pi}$. In addition, we can simplify the upper bound to $\mathbb{E}[L_{cs}(M,\lambda)]$ from Remark~\ref{UnivUBRmk} under the noise limited scenario, which is given as follows:
	\begin{align}\label{NoiseLimitedLemmaEq3}
	\allowdisplaybreaks
	&\mathbb{E}[L_{cs}(M,\lambda)]\nonumber\\ 
	\overset{(a)}{\leq}&\int_{0}^{\infty}\exp\biggl(\frac{W\Gamma_{cs}\mathit{l}(r_0)}{PM}\biggl) \lambda 2\pi r_0 \exp(-\lambda \pi r_0^2) {\rm d}r_0 \nonumber\\
	=&\int_{0}^{R_c}\exp\biggl(\frac{W\Gamma_{cs}C_L r_0^{\alpha_L}}{PM}\biggl) \lambda 2\pi r_0 \exp(-\lambda \pi r_0^2) {\rm d}r_0+ \int_{R_c}^{\infty}\exp\biggl(\frac{W\Gamma_{cs}C_N r_0^{\alpha_N}}{PM}\biggl) \lambda 2\pi r_0 \exp(-\lambda \pi r_0^2) {\rm d}r_0\nonumber\\
	<&\exp\biggl(\frac{W\Gamma_{cs}C_L R_c^{\alpha_L}}{PM}\biggl) \biggl(1-\exp(-\lambda \pi R_c^2)\biggl)+ \int_{R_c}^{\infty}\exp\biggl(\frac{W\Gamma_{cs}C_N r_0^{\alpha_N}}{PM}\biggl) \lambda 2\pi r_0 \exp(-\lambda \pi r_0^2) {\rm d}r_0,
	\end{align}
	where (a) is obtained by applying the noise limited assumption to (\ref{CSSuccProbEq1}), and noting that the BS providing the smallest path loss among all the BSs is the closest BS of $\Phi$ to the origin. Since $\alpha_N =2$, it can be observed from~(\ref{NoiseLimitedLemmaEq3}) that $\mathbb{E}[L_{cs}(M)]$ is guaranteed to have a finite mean if $\lambda M > \frac{\Gamma_{cs}C_N W}{P\pi}$. 
\end{proof}

We can observe from the proof of Lemma~\ref{NoiseLmdPLE2Lem} that for any arbitrarily large distance $r_0$, there is a fraction $\exp(-\lambda \pi r_0^2)$ of cell edge users whose nearest BSs are farther than $r_0$, and the number of IA cycles for these edge users to succeed  cell search scales as $\exp(\frac{W\Gamma_{cs}C_N r_0^{2}}{PM})$. As a result, if the BS deployment is too sparse or the number of BS antennas/beams is such that $\lambda M< \frac{\Gamma_{cs}C_N W}{P\pi}$, the cell search delay averaged over all the users becomes infinity due to cell edge users. By contrast, with network densification, the fraction of cell edge users with poor signal power is reduced, and the average cell search delay can be reduced to a finite mean value whenever  $\lambda M> \frac{\Gamma_{cs}C_N W}{P\pi}$. A similar behavior happens when the BSs are using more antennas to increase the SNR for the cell edge users.

To summarize, for the noise limited scenario such as a mmWave network, 
the mean cell search delay is infinite whenever the NLOS path loss exponent $\alpha_N >2$, which is typically the case. However, for the special case with NLOS path loss exponent $\alpha_N = 2$, the mean cell search delay could switch from infinity to a finite value through careful network design, such as BS densification or adopting more BS antennas.
\subsubsection{Interference limited Scenario}\label{IntLmtSubSec} In the interference limited scenario, the noise power is dominated by the interference power, so that we can assume $W = 0$. For example, a massive MIMO network that operates in the sub-6 GHz bands is typically interference limited~\cite{Marzetta2010noncooperative}. In this part, we investigate the cell search delay in an interference-limited network with a standard single slope path loss function $\mathit{l}(r) = C r^\alpha$, which is suitable for networks with sparsely deployed BSs as opposed to ultra-dense networks~\cite{zhang2015downlink}. 

First, we prove that Theorem~\ref{ExpCSProbThm} can be greatly simplified under this interference limited scenario. 
\begin{lemma}\label{ExpectedCSIntLmtdLemma}
	Under the interference limited scenario, the expected number of initial access cycles required to succeed in cell search  is given by:
	\begin{align}\label{ExpectedCSIntLmtdEq}
	\mathbb{E}[L_{cs}(M)] = \sum_{j=0}^{\infty} \bigg(\sum_{k=0}^{j} \frac{(-1)^k {j \choose k}}{1+2\int_{1}^{+\infty} (1-(1+\Gamma_{cs}/r^\alpha)^{-k})r {\rm d}r}\bigg)^M.
	\end{align}
\end{lemma}
\begin{proof}
	By substituting $W = 0$ and $\mathit{l}(r) = C r^\alpha$ into~(\ref{ExpCSProbThmEq2}), $A_j$ defined in~(\ref{ExpCSProbThmEq2}) can be further simplified as follows:
	\begin{align*}
	\allowdisplaybreaks
		A_j = &\int_{0}^{\infty}  \biggl\{\sum\limits_{k=0}^{j} (-1)^k {j \choose k}  \exp\biggl(-\frac{2\pi\lambda}{M} \int_{r_1}^{\infty} \bigg(1-\frac{1}{(1+\Gamma_{cs} r_1^\alpha/r^\alpha)^k}\bigg)r {\rm d}r\biggr)\biggl\} \frac{2\lambda \pi r_1}{M} \exp(-\frac{\lambda \pi r_1^2}{M}) {\rm d}r_1 \nonumber\\
		=&\sum\limits_{k=0}^{j} (-1)^k {j \choose k} \biggl\{\int_{0}^{\infty} \exp\biggl(-\frac{2\pi\lambda r_1^2}{M} \int_{1}^{\infty} \bigg(1-\frac{1}{(1+\Gamma_{cs} /r^\alpha)^k}\bigg)r {\rm d}r\biggr)\frac{2\lambda \pi r_1}{M} \exp(-\frac{\lambda \pi r_1^2}{M}) {\rm d}r_1 \biggl\}  \nonumber\\
		=& \sum\limits_{k=0}^{j} \frac{(-1)^k {j \choose k}}{1+2\int_{1}^{+\infty} (1-(1+\Gamma_{cs}/r^\alpha)^{-k})r {\rm d}r},
	\end{align*}
	which completes the proof.
\end{proof}

\begin{remark}\label{ExpectedCSIntLmtdRmk}
	We can observe from Lemma~\ref{ExpectedCSIntLmtdLemma} that $\mathbb{E}[L_{cs}(M)]$ does not depend on the BS intensity $\lambda$ under the interference limited scenario. This is because the increase and decrease of the signal power can be perfectly counter-effected by the corresponding increase and decrease of the interference power~\cite{trac}. Another immediate observation from Lemma~\ref{ExpectedCSIntLmtdLemma} is that $A_j$ is independent of the number of BS antennas $M$ for $\forall j$. Since $A_j \leq 1$ according to its definition in Theorem~\ref{ExpCSProbThm}, $\mathbb{E}[L_{cs}(M)]$ is therefore monotonically non-increasing with respect to $M$, which is a stronger observation than Lemma~\ref{BeamSweepGeneralLemma}. 
\end{remark}

\begin{remark}\label{ExpectedCSIntLmtdPLE2Rmk}
	If the path loss exponent $\alpha = 2$, it can be proved from Lemma~\ref{ExpectedCSIntLmtdLemma} that $\mathbb{E}[L_{cs}(M)] = \infty$ for $\forall M$. This is mainly because the interference power will dominate the signal power when $\alpha = 2$, so that the coverage probability is 0 for any SINR threshold $\Gamma_{cs}$. 
\end{remark}

If $\alpha>2$, we can prove that there may exist a phase transition for $\mathbb{E}[L_{cs}(M)]$ in terms of the BS beam number $M$. In order to show that, we first apply Remark~\ref{UnivUBRmk} and obtain a sufficient condition to guarantee the finiteness for $\mathbb{E}[L_{cs}(M)]$.  
\begin{lemma}\label{IntLmtdLemma2}
	Under the interference limited scenario with path loss exponent $\alpha>2$, the expected number of IA cycles to succeed cell search is such that $\mathbb{E}[L_{cs}(M)] < \infty$ if the number of BS beams is such that $M > \frac{2\Gamma_{cs}}{\alpha-2}$, where $\Gamma_{cs}$ denotes the detection threshold for a BS. In particular, when $M = 1$, i.e., the BS is omni-directional, $\mathbb{E}[L_{cs}(1)]$ is finite if and only if $\alpha > 2\Gamma_{cs} + 2$.
\end{lemma} 

\begin{proof}
	Denote by $x_0$ the closest BS to the origin among $\Phi$, and $S_M(i^*)$ the BS sector containing $x_0$, we can obtain an upper bound to $\mathbb{E}[L_{cs}(M)]$ by substituting~(\ref{CSSuccProbEq1}) and $W = 0$ into Remark~\ref{UnivUBRmk} as follows:
	\begin{align}\label{IntLmtdLemPfEq}
	\allowdisplaybreaks
	\mathbb{E}[L_{cs}(M)] &\leq \mathbb{E} \biggl[\prod_{x_j \in \Phi \cap S_M(i^*) \setminus \{x_0\}} \biggl(1+\Gamma_{cs} \mathit{l}(\|x_0\|)/\mathit{l}(\|x_j\|)\biggl)\biggl] \nonumber\\
	&\overset{(a)}{=}\int_{0}^{\infty} \mathbb{E}\biggl[\prod_{x_j \in \Phi \cap S_M(i^*) \cap B^c(o,r_0)} \biggl(1+\Gamma_{cs} \mathit{l}(r_0)/\mathit{l}(\|x_j\|)\biggl)\biggl] 2\lambda \pi r_0 \exp(-\lambda \pi r_0^2) {\rm d}r_0 \nonumber\\
	&\overset{(b)}{=}\int_{0}^{\infty}\exp\biggl(\frac{2\pi\lambda \Gamma_{cs}}{M} \int_{r_0}^{\infty} \frac{\mathit{l}(r_0) r}{\mathit{l}(r)} {\rm d}r\biggl) 2\lambda \pi r_0 \exp(-\lambda \pi r_0^2) {\rm d}r_0 \nonumber\\
	&\overset{(c)}{=}\int_{0}^{\infty}  \exp\biggl(-\left(1-\frac{2 \Gamma_{cs} }{M (\alpha-2)}\right)v \biggl) {\rm d}v\nonumber\\
	&=\left\{
	\begin{array}{ll}
	\infty, & \text{if } M \leq \frac{2\Gamma_{cs}}{\alpha-2}, \\
	\frac{M(\alpha-2)}{M(\alpha-2)-2\Gamma_{cs}}, & \text{if } M > \frac{2\Gamma_{cs}}{\alpha-2},
	\end{array}
	\emph{ } \right.
	\end{align}
	where (a) is obtained by noting that $x_0$ is the closest BS to the origin, (b) follows from the PGFL for the PPP\footnote{Note that~\cite[Theorem 4.9]{stochgeom} does not directly apply to the PGFL calculation here since $f(x) = 1+\Gamma_{cs} \mathit{l}(r_0)/\mathit{l}(x)$ is larger than 1. However, we can use dominated convergence theorem to prove that for PPP $\Phi$ with intensity measure $\Lambda(\cdot)$, the PGFL result still holds if function $f(x)$ satisfies $f(x) \geq 1$ and $\int_{\mathbb{R}^2}(f(x) -1)\Lambda({\rm d}x) < \infty$, i.e. $\mathbb{E}[\prod_{x_i \in \Phi} f(x_i)] = \exp(\int_{\mathbb{R}^2}(f(x)-1)\Lambda({\rm d}x)$.}, and (c) is derived through change of variables (i.e. $v = \lambda \pi r_0^2$). It can be observed that~(\ref{IntLmtdLemPfEq}) is finite whenever $M > \frac{2\Gamma_{cs}}{\alpha-2}$, which is a sufficient condition for the finiteness of $\mathbb{E}[L_{cs}(M)]$. In particular, the equality holds in the first step of~(\ref{IntLmtdLemPfEq}) when $M=1$. As a result, $\mathbb{E}[L_{cs}(1)]$ is finite if and only if $\alpha > 2\Gamma_{cs} + 2$.
\end{proof}

According to Lemma~\ref{ExpectedCSIntLmtdLemma} and Lemma~\ref{IntLmtdLemma2}, the number of IA cycles to succeed in cell search (i.e., $\mathbb{E}[L_{cs}(M)]$) may have a phase transition in terms of the number of BS beams $M$, depending on the relation between the path loss exponent $\alpha$ and the detection threshold $\Gamma_{cs}$. This is detailed in the following theorem.
\begin{theorem}\label{PhaseTranIntLmtedRmk} The number of IA cycles to succeed in cell search for the interference limited networks satisfy the following:
	\begin{itemize}
		\item If $\alpha > 2 + 2\Gamma_{cs}$, $\mathbb{E}[L_{cs}(M)] < \infty$ for the omni-directional BS antenna case, i.e., $M=1$. By the monotonicity of $\mathbb{E}[L_{cs}(M)]$ with respect to $M$, $\mathbb{E}[L_{cs}(M)]$ is guaranteed to be finite for any $M \geq 1$. 
		\item If $\alpha \leq 2 + 2\Gamma_{cs}$, $\mathbb{E}[L_{cs}(M)] = \infty$ for $M=1$, and $\mathbb{E}[L_{cs}(M)] < \infty$ if $M > \frac{2\Gamma_{cs}}{\alpha-2}$. Therefore, according to the monotonicity of $\mathbb{E}[L_{cs}(M)]$, there exists a phase transition at $M^* \in [2,\frac{2\Gamma_{cs}}{\alpha-2}]$, such that $\mathbb{E}[L_{cs}(M)] = \infty$ for $M \leq M^*$, and $\mathbb{E}[L_{cs}(M)] < \infty$ for $M > M^*$. In particular,  $\mathbb{E}[L_{cs}(M)] = \infty$ for $\forall M$ if  $\alpha = 2$, which means $M^* = \infty$.
	\end{itemize}
\end{theorem}

The path loss exponent $\alpha$ depends on the propagation environment, and $\alpha=2$ corresponds to a free space LOS scenario; while $\alpha$ increases as the environment becomes relatively more lossy and scatter-rich, such as urban and suburban areas. In addition, the SINR detection threshold $\Gamma_{cs}$ depends on the receiver decoding capability, which is typically within $-10$ dB and $0$ dB~\cite{barati2015directional}. Theorem~\ref{PhaseTranIntLmtedRmk} shows that in a lossy environment with $\alpha > 2+2\Gamma_{cs}$, the typical user can detect a nearby BS in a finite number of IA cycles on average. This is mainly because the relative strength of the useful signal with respect to the interfering signals is strong enough. However, when $\alpha \leq 2+2\Gamma_{cs}$, $\mathbb{E}[L_{cs}(M)]$ could be infinite due to the significant fraction of cell edge users that have poor SIR coverage and therefore require a very high number of IA cycles to succeed in cell search. Specifically, when $M$ is very small (e.g., $M=1$), the edge user is subject to many strong nearby interferers inside every BS sector, so that the corresponding cell search delay averaged over all users becomes
infinity. However, as $M$ increases, the BS beam sweeping will create enough angular separation so that the nearby BSs to the edge user could locate in different BS sectors. As a result, $L_{cs}(M)$ is significantly decreased for cell edge users as $M$ increases, and therefore the phase transition for $\mathbb{E}[L_{cs}(M)]$ happens. 

In summary, for an interference-limited network, we can always ensure the network to be in a desirable condition with finite mean cell search delay by tuning the number of BS beams/antennas $M$ appropriately. 

\subsection{Cell Search Delay Distribution in Poisson Networks with Rayleigh Fading}
The previous part is mainly focused on the mean number of IA cycles to succeed in cell search $\mathbb{E}[L_{cs}(M,\lambda)]$, or equivalently the mean cell search delay. However, as shown in Theorem~\ref{NoiseLimLemma1}, Theorem~\ref{NoiseLmdPLE2Lem} and Theorem~\ref{PhaseTranIntLmtedRmk}, $\mathbb{E}[L_{cs}(M,\lambda)]$ could be infinite under various settings, and there are large variations of the performance between cell edge user and cell center user. Therefore, it is also important to analyze the cell search delay distribution for system design. 

Since the cell search delay $D_{cs}(M,\lambda)$ depends on the spatial point process model for BSs and the fading random variables at each IA cycle, its distribution is intractable in general. In this section, we evaluate the distribution of the conditional mean cell search delay given the distance from the typical user to its closest BS $R_0$, which is a random variable with PDF $f_{R_0}(r_0) = 2 \pi \lambda r_0\exp(-\lambda \pi r_0^2)$. Specifically, 
we first derive the expected number of IA cycles to succeed in cell search given $R_0$, i.e., $\mathbb{E}[L_{cs}(M,\lambda)|R_0]$, which is a function of random variable $R_0$ with mean $\mathbb{E}[L_{cs}(M,\lambda)]$. For notation simplicity, we denote by $L_{cs}(R_0,M,\lambda) \triangleq \mathbb{E}[L_{cs}(M,\lambda) |R_0]$ for the rest of the paper. According to~(\ref{CSandIADelayDefnEq}), we will evaluate the distribution of the following conditional mean cell search delay: 
\begin{align}\label{CondiCSDelayDefnEq}
D_{cs}(R_0,M,\lambda) \triangleq (L_{cs}(R_0,M,\lambda)-1) T + M \tau.
\end{align}The main reason to investigate the cell search delay conditionally on $R_0$ is because $R_0$ captures the location and therefore the signal quality of the typical user. In particular, $R_0 \ll \frac{1}{2\sqrt{\lambda}}$ corresponds to the cell center user, while $R_0 \gg \frac{1}{2\sqrt{\lambda}}$ corresponds to the cell edge user, where $\frac{1}{2\sqrt{\lambda}}$ represents the mean distance from the typical user to its nearest BS on the PPP $\Phi$. 

In order to derive $L_{cs}(R_0,M,\lambda)$ in~(\ref{CondiCSDelayDefnEq}), we will first derive $\mathbb{E}[L_{cs}(M,\lambda)|R_1,R_2,...,R_M]$, where $R_i$ denotes the distance from the typical user to its closest BS in the $i$-th BS sector (i.e., $R_i = \|x_0^i\|$) for $1 \leq i \leq M$. 
\begin{lemma}\label{MeanCSGivenAllDistsLemma}
	Given the distances from the typical user to its nearest BSs inside every BS sector $R_1, ...,R_M$, the mean number of IA cycles for cell search is:
	\begin{align}\label{MeanCSGivenAllDistsEq1}
	\mathbb{E}[L_{cs}(M,\lambda)|R_1,R_2,...,R_M] = \sum_{j=0}^{\infty} \prod_{i=1}^{M} f_j(R_i,M,\lambda),
	\end{align}
	where $f_j(R_i)$ denotes the probability that $x_0^i$ is detected in the first $j$ IA cycles, which is: 
	\begin{align*}
	\!\!\!\!\!\!f_j(R_i,M,\lambda) = \sum_{k=0}^{j} (-1)^k {j \choose k} \exp\left(- \frac{Wk\Gamma_{cs} \mathit{l}(R_i)}{PM}\right) \exp\left(-\frac{2\lambda \pi}{M} \int_{R_i}^{\infty} (1-\frac{1}{(1+\Gamma_{cs}\mathit{l}(R_i)/\mathit{l}(r))^k})r{\rm d}r\right).
	\end{align*}
\end{lemma}
\begin{proof}
	We can first prove $\mathbb{E}[L_{cs}(M,\lambda)|R_1,R_2,...,R_M] = \mathbb{E}[\mathbb{E}[L_{cs}(M,\lambda)| \Phi] | R_1,R_2,...,R_M]$, which is due to the tower property for conditional expectations. The rest of the proof follows steps similar to those of Theorem~\ref{ExpCSProbThm}, and therefore we omit the details. 
\end{proof}

Next we prove the following corollary to derive $L_{cs}(R_0,M,\lambda)$ from $\mathbb{E}[L_{cs}(M,\lambda)|R_1,R_2,...,R_M]$. 
\begin{corollary}\label{CoroExpMRVs}
	For all i.i.d. non-negative random variables $R_1$, $R_2$,...,$R_M$ with CCDF $G(r)$, and all functions $F: [0,\infty)^M \rightarrow [0,\infty)$ which are symmetric, the following relation holds true:
	\begin{align}\label{CoroExpMRVsEq}
	\mathbb{E}[F(R_1,R_2,...,R_M) | \min(R_1,R_2,...,R_M) = r] = \frac{\mathbb{E}[F(r,R_2,...,R_M)\mathbbm{1}_{\{R_j > r, \forall j \neq 1\}}]}{(G(r))^{M-1}}.
	\end{align}
\end{corollary}
\begin{proof}
	Denote by $R_0 = \min(R_1,R_2,...,R_M)$, then we can obtain~(\ref{CoroExpMRVsEq}) as follows:
	\allowdisplaybreaks
	\begin{align*}
	\allowdisplaybreaks
	&\mathbb{E}[F(R_1,R_2,...,R_M) | R_0 = r] \\
	=&\lim\limits_{\epsilon \rightarrow 0}\frac{\mathbb{E}[F(R_1,R_2,...,R_M) \times \mathbbm{1}_{|R_0-r|<\epsilon}]}{\mathbb{P}(|R_0-r| < \epsilon)}\\
	=&\lim\limits_{\epsilon \rightarrow 0} \frac{\sum_{i=1}^{M} \mathbb{E}[F(R_1,R_2,...,R_M) \times \mathbbm{1}_{(\{|R_i-r|<\epsilon\}\cap \{R_j > R_i, \forall j \neq i \})}]}{\sum_{k=1}^{M}\mathbb{P}(\{|R_k-r|<\epsilon\}\cap \{R_j > R_k, \forall j \neq k \})}\\
	=&\lim\limits_{\epsilon \rightarrow 0} \frac{\sum_{i=1}^{M} \mathbb{E}[F(R_1,R_2,...,R_M) \mathbbm{1}_{(\{R_j > R_i, \forall j \neq i \})} | |R_i-r|<\epsilon]}{\sum_{k=1}^{M}\mathbb{P}(\{R_j > R_k, \forall j \neq k \} | |R_k-r|<\epsilon)}\\
	=& \frac{\sum_{i=1}^{M} \mathbb{E}[F(R_1,R_2,...,R_M) \mathbbm{1}_{(\{R_j > R_i, \forall j \neq i \})} | R_i=r]}{\sum_{k=1}^{M}\mathbb{P}(\{R_j > R_k, \forall j \neq k \} | R_k=r)},
	\end{align*}
	the proof is completed by noting $F$ is symmetric. 
\end{proof}

By taking $F(R_1,R_2,...,R_M) = \mathbb{E}[L_{cs}(M,\lambda)|R_1,R_2,...,R_M]$ in Corollary~\ref{CoroExpMRVs}, $\mathbb{E}[L_{cs}(M,\lambda)|R_0]$ can directly obtained as follows.
\begin{lemma}\label{MeanCSGivenNNDistsLemma}
	Given the distance from the typical user to the nearest BS $R_0$, the mean number of IA cycles to succeed cell search is:
	\begin{align*}
	\!\!\!\!L_{cs}(R_0,M,\lambda)= \sum_{j=0}^{\infty} f_j(R_0,M,\lambda) \biggl\{\int_{R_0}^{\infty} f_j(r,M,\lambda) \frac{\lambda 2\pi r}{M} \exp(-\frac{\lambda \pi r^2}{M}) {\rm d}r\biggl\}^{M-1}\!\!\!\!\!\!\exp\biggl(\frac{\lambda\pi(M-1)R_0^2}{M}\biggl),
	\end{align*}
	where the function $f_j(r,M,\lambda)$ is defined in Lemma~\ref{MeanCSGivenAllDistsLemma}. 
\end{lemma}

Lemma~\ref{MeanCSGivenNNDistsLemma} provides a method to evaluate the cell search delay distribution under a general setting. For noise limited networks and interference limited networks, we can obtain the following simplified results.
\begin{corollary}\label{CSDelayDistNoiseLmtedCoro}
For the noise limited network, $L_{cs}(R_0,M,\lambda)$ is given by:
\allowdisplaybreaks
\begin{align}
\allowdisplaybreaks
\!\!\!\!L_{cs}(R_0,M,\lambda)=
\left\{
\begin{array}{ll}
\sum_{j=0}^{\infty} (1-\exp(-\frac{\Gamma_{cs}WC_N R_0^{\alpha_N}}{PM}))^j \{\int_{R_0}^{\infty}(1-\exp(-\frac{\Gamma_{cs}WC_N r^{\alpha_N}}{PM}))^j\\
\times\frac{\lambda 2\pi r}{M} \exp(-\frac{\lambda \pi r^2}{M}) {\rm d}r\}^{M-1} \exp(\lambda \pi \frac{M-1}{M}R_0^2), &\text{if } R_0\geq R_c, \\
\sum_{j=0}^{\infty} (1-\exp(-\frac{\Gamma_{cs}WC_L R_0^{\alpha_L}}{PM}))^j \{\int_{R_C}^{\infty}(1-\exp(-\frac{\Gamma_{cs}WC_N r^{\alpha_N}}{PM}))^j\\
\times\frac{\lambda 2\pi r}{M} \exp(-\frac{\lambda \pi r^2}{M}) {\rm d}r + \int_{R_0}^{R_C}(1-\exp(-\frac{\Gamma_{cs}WC_L r^{\alpha_L}}{PM}))^j\\
\times\frac{\lambda 2\pi r}{M} \exp(-\frac{\lambda \pi r^2}{M}) {\rm d}r\}^{M-1} \exp(\lambda \pi \frac{M-1}{M}R_0^2), &\text{if } R_0< R_c.
\end{array}
\emph{ } \right.
\end{align}
\end{corollary}
Corollary~\ref{CSDelayDistNoiseLmtedCoro} can be easily proved from Lemma~\ref{MeanCSGivenNNDistsLemma} and the fact that interference power is 0. 

\begin{corollary}\label{CSDelayDistIntLmtedCoro}
	For the interference limited network and the standard single-slope path loss model with path loss exponent $\alpha >2$, $L_{cs}(R_0,M,\lambda)$ is given by:
	\allowdisplaybreaks
	\begin{align}
	\allowdisplaybreaks
	L_{cs}(R_0,M,\lambda) =& \sum_{j=0}^{\infty} \biggl\{\sum_{k=0}^{j}(-1)^k {j \choose k} \exp\biggl(-\frac{2\pi\lambda R_0^2 H(k,\alpha,\Gamma_{cs})}{M}\biggl) \biggl\} \nonumber\\
	&\times \biggl\{\sum_{k=0}^{j} \frac{(-1)^k{j \choose k} \exp\biggl(-\frac{2\pi\lambda R_0^2 H(k,\alpha,\Gamma_{cs})}{M}\biggl)}{1+2H(k,\alpha,\Gamma_{cs})}\biggl\}^{M-1},
	\end{align}
	where $H(k,\alpha,\Gamma_{cs}) = \int_{1}^{\infty} (1-\frac{1}{(1+\Gamma/r^\alpha)^k}) r {\rm d}r$.
\end{corollary}

\begin{proof}
	Since $W=0$ and $\mathit{l}(r) = C r^\alpha$, $f_j(R_i,M,\lambda)$ in Lemma~\ref{MeanCSGivenNNDistsLemma} can be simplified as:
	\begin{align}\label{CSDelayDistIntLmtedCoroPfEq1}
	f_j(R_i,M,\lambda) = \sum_{k=0}^{j}(-1)^k {j \choose k} \exp\biggl(-\frac{2\pi\lambda R_0^2 H(k,\alpha,\Gamma_{cs})}{M}\biggl).
 	\end{align}
 	Therefore, we can further obtain that: 
	\begin{align}\label{CSDelayDistIntLmtedCoroPfEq2}
	\int_{R_0}^{\infty}f_j(r,M,\lambda)\frac{\lambda 2\pi r}{M} \exp(-\frac{\lambda \pi r^2}{M}) {\rm d}r =  \sum_{k=0}^{j}(-1)^k {j \choose k}\frac{\exp(-\frac{\lambda\pi}{M}(1+2H(k,\alpha,\Gamma_{cs}))R_0^2)}{1+2H(k,\alpha,\Gamma_{cs})}.
	\end{align}
	The proof can be completed by substituting~(\ref{CSDelayDistIntLmtedCoroPfEq1}) and~(\ref{CSDelayDistIntLmtedCoroPfEq2}) into Lemma~\ref{MeanCSGivenNNDistsLemma}.
\end{proof}


\section{Numerical Evaluations}\label{NumEval}
In this section, the distribution of the conditional mean cell search delay~(\ref{CondiCSDelayDefnEq}) is numerically evaluated for both the noise limited scenario and the interference limited scenario. Specifically, for the noise limited scenario, we consider a cellular network operating in the mmWave band with carrier frequency $f_c = 73$ GHz, bandwidth $B = 2$ GHz, and BS intensity $\lambda = 100$ BS/km$^2$. The path loss exponents for LOS and NLOS links are $2.1$ and $3.3$ respectively, and the critical distance is $R_c = 50$m. In addition, the OFDM symbol period is $\tau = 14.3$ $\mu$s, and the IA cycle length is chosen as $T = 20$ ms~\cite{li2016design,Verizon20165G2}. As for the interference limited scenario, we consider a cellular network with carrier frequency $f_c = 2$ GHz, BS intensity $\lambda = 100$ BS/km$^2$, and a standard single slope path loss model with path loss exponent $\alpha = 2.5$. The OFDM symbol period is $\tau = 71.4$ $\mu$s, and the IA cycle length is $T = 100$ ms.

\begin{figure}
	\begin{subfigure}[b]{0.50\textwidth}
		\centering
		\includegraphics[height=2.35in, width= 3.6in]{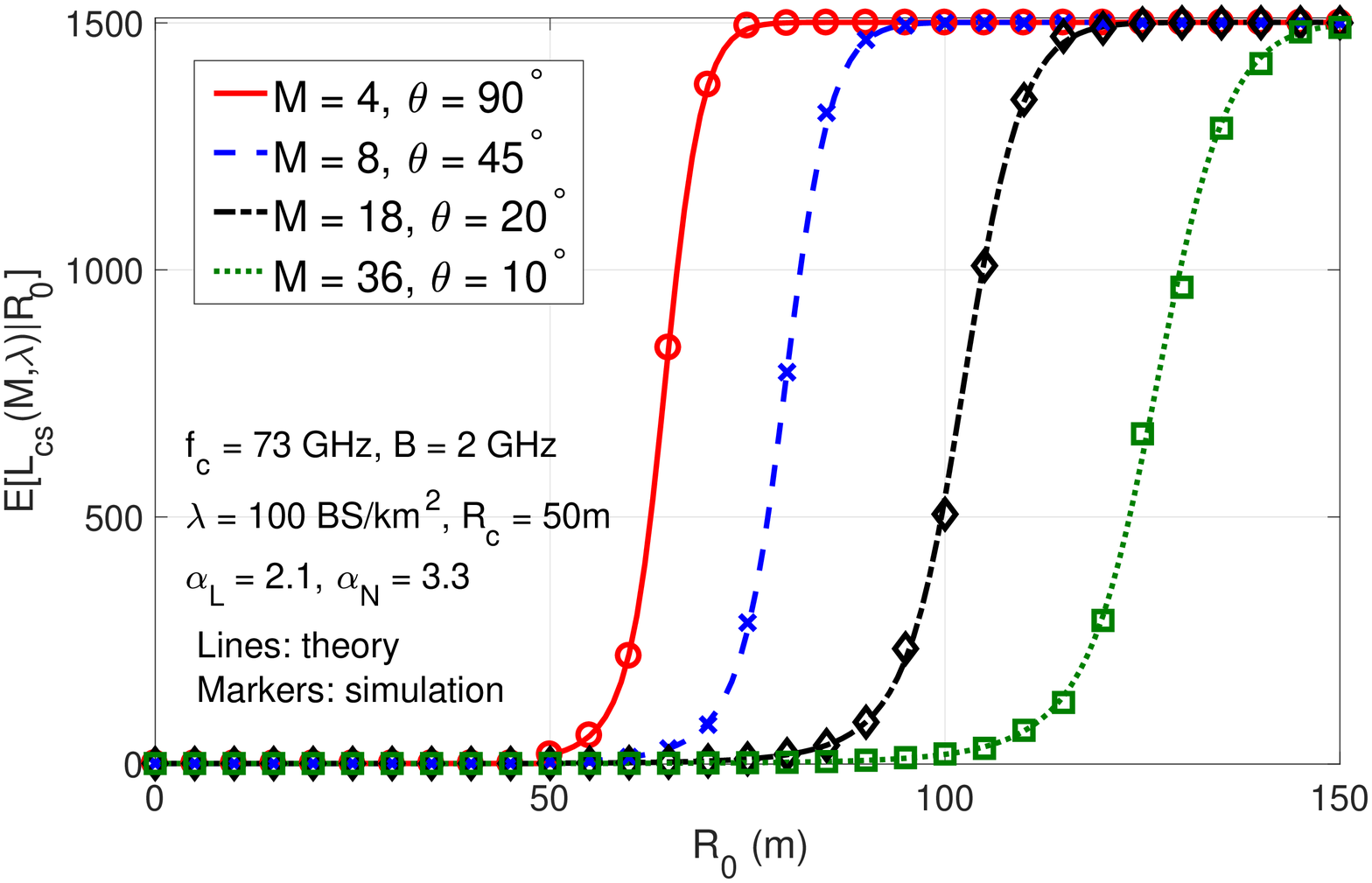}
		\caption{Noise limited networks}
		\label{NoiseLmdNtw}
	\end{subfigure}
	\hfill
	\begin{subfigure}[b]{0.51\textwidth}
		\centering
		\includegraphics[height=2.35in, width=3.6in]{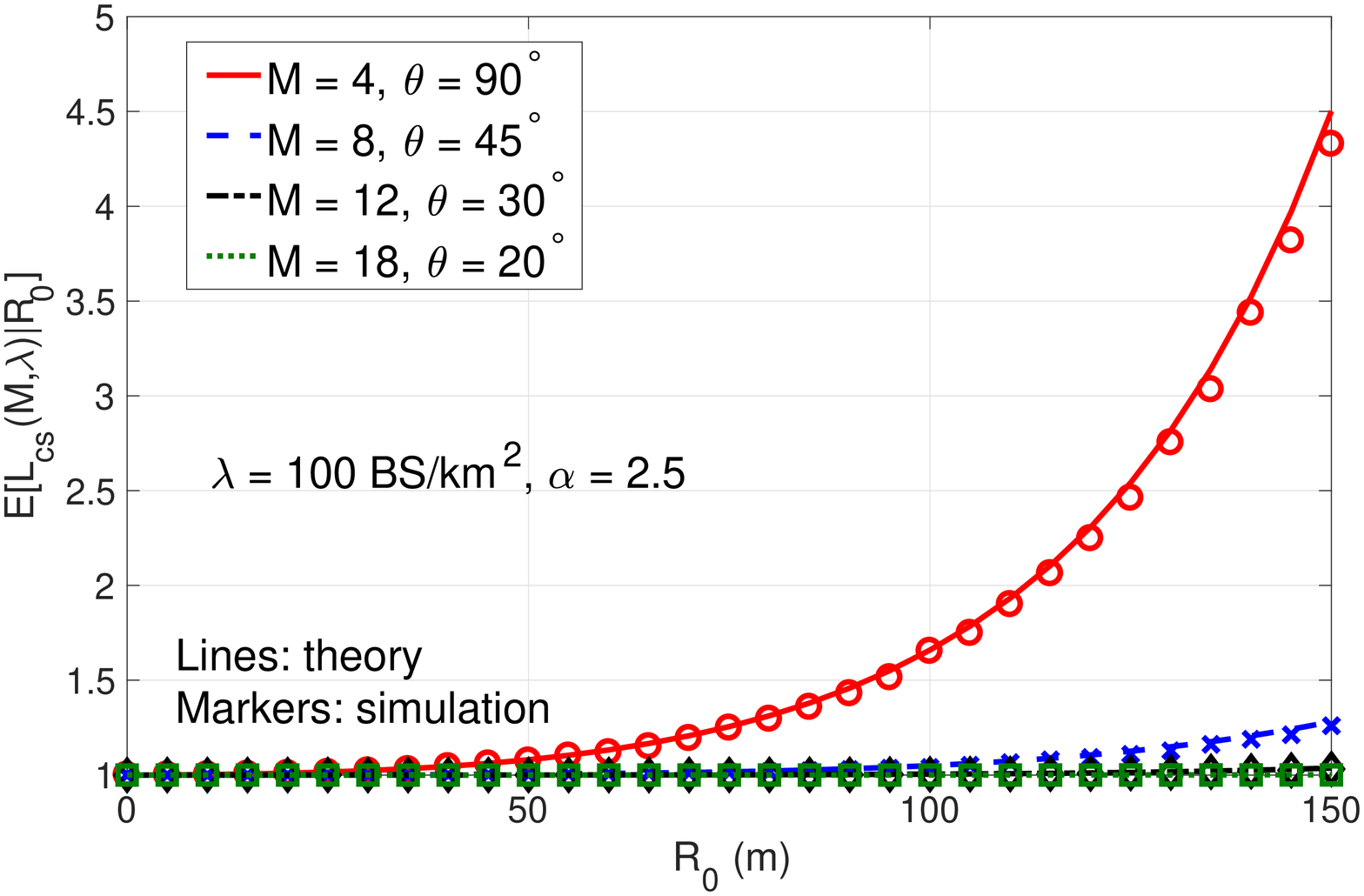}
		\caption{Interference limited networks}
		\label{IntLmdNtw}
	\end{subfigure}
	\caption{Conditional expected number of cycles to succeed in cell search.}\label{NumCycsFig}
\end{figure}
\subsection{Conditional Expected Number of Cycles to Succeed in Cell Search}
In order to evaluate the distribution of the conditional mean cell search delay, we first illustrate Lemma~\ref{MeanCSGivenNNDistsLemma}. Specifically, we have simulated the cellular network with the directional cell search protocol proposed in Section~\ref{IAProtSubsec}, given the distance from the user to its nearest BS $R_0$. As shown in Lemma~\ref{NoiseLimLemma1} and Remark~\ref{PhaseTranIntLmtedRmk}, the cell edge users will require a large number of cycles to succeed in cell search. Therefore, we  have set an upper bound for the number of cycles that a user can try cell search, which is equal to $1500$ cycles for the noise limited scenario and $100$ cycles for the interference limited scenario. Specifically, the infinite summation in Lemma~\ref{NoiseLimLemma1} is computed up to the $1500$-th ($100$-th) term, and the simulation will treat a user as in outage if it cannot be connected within $1500$ ($100$) cycles. 

Fig.~\ref{NumCycsFig} shows a close match between the analytical results and simulation results for both the noise and interference limited scenarios, which is in line with Lemma~\ref{MeanCSGivenNNDistsLemma}. In addition, we can also observe from Fig.~\ref{NumCycsFig} that the conditional expected number of cycles to succeed in cell search is monotonically decreasing as the number of BS antennas/beams $M$ increases, or as the distance to the nearest BS $R_0$ decreases.  
 
\subsection{Cell Search Delay Distribution in Noise Limited Networks}
The cell search delay distribution for noise limited networks is numerically evaluated in this part. Fig.~\ref{CSDelayDistFig} plots the CCDF of the conditional mean cell search delay $D_{cs}(R_0,M,\lambda)$, which is obtained by generating $10^6$ realizations of $R_0$ and computing the corresponding $D_{cs}(R_0,M,\lambda)$ through Corollary~\ref{CSDelayDistNoiseLmtedCoro}. We can observe from Fig.~\ref{CSDelayDistFig} that under the log-log scale, the tail distribution function of $D_{cs}(R_0,M,\lambda)$, i.e., $\mathbb{P}(D_{cs}(R_0,M,\lambda) \geq t)$, decreases almost linearly with respect to $t$. This indicates that the cell search delay is actually heavy-tailed and of the Pareto type. It can also be observed from Fig.~\ref{CSDelayDistFig} that the tail distribution function satisfies $\lim_{t \rightarrow \infty} \frac{-\log \mathbb{P}(D_{cs}(R_0,M,\lambda) \geq t)}{\log t} < 1$ for $M = 4,8,18,36$. Therefore, the expected cell search delay is always infinite, which is in line with Lemma~\ref{NoiseLimLemma1}.

Fig.~\ref{CSDelayDistFig} also shows that as the number of BS antennas $M$ increases, the tail of $D_{cs}(R_0,M,\lambda)$ becomes lighter and thus the cell search delay for edge users is significantly reduced. For example, the cell search delay for the $10^{\text{th}}$ percentile user is almost 10 times smaller when $M$ increases from $18$ to $36$. In fact, increasing $M$ will increase the SNR of cell edge users, such that the number of IA cycles required for the edge users to succeed in cell search (i.e., $L_{cs}(R_0,M,\lambda)$) can be shortened. Since $D_{cs}(R_0,M,\lambda) \triangleq T(L_{cs}(R_0,M,\lambda)-1) + M \tau$, and the IA cycle length $T$ is much larger than the OFDM symbol period $\tau$, the tail distribution of $D_{cs}(R_0,M,\lambda)$ therefore becomes lighter as $M$ increases despite having a higher beam-sweeping overhead within every IA cycle. 
\begin{figure}[h]
	\centering
	\includegraphics[width=0.55\linewidth]{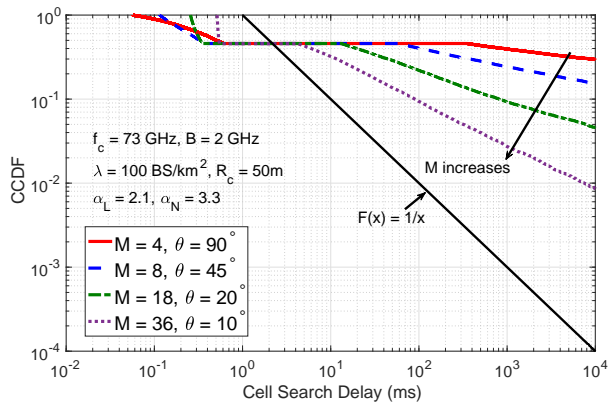}
	\caption{Cell search delay distribution for noise limited networks.}
	\label{CSDelayDistFig}
\end{figure}

Due to the heavy-tailed nature for the cell search delay distribution, Fig.~\ref{CSDelayDistFig} shows that there exists an extremely large variation of the cell search delay performance from cell center users to cell edge users. Fig.~\ref{5thDelayFigNoiseLmted} plots the cell search delay for the $95^{\text{th}}$ percentile users, as the number of BS antennas $M$ increases. Since the $95^{\text{th}}$ percentile users are located at the cell center, they are typically LOS to their serving BSs with sufficiently high isotropic SNR, and thus they can succeed cell search in the first cycle that they initiates IA. Therefore, Fig.~\ref{5thDelayFigNoiseLmted} shows that as $M$ increases, the cell search delay for the $95^{\text{th}}$ percentile users increases almost linearly due to the increase of the beam-sweeping overhead.
\begin{figure}[h]
	\centering
	\includegraphics[width=0.55\linewidth]{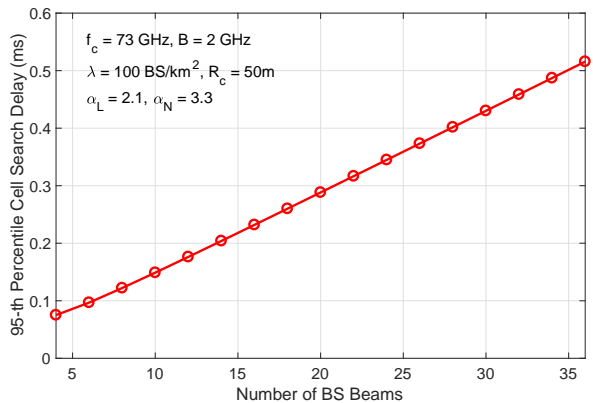}
	\caption{$95^{\text{th}}$ percentile cell search delay for noise limited network.}
	\label{5thDelayFigNoiseLmted}
\end{figure}

The cell search delay performance for the $50^{\text{th}}$ percentile users, or the median users, is plotted in Fig.~\ref{50thDelayFigNoiseLmted}. We can observe that in contrast to the mean cell search delay which is infinite, the median delay is less than 1 ms for various BS antenna number $M$. When $M$ is small, median users do not have high enough SNR and thus they will need more than $1$ IA cycles to succeed in cell search. As $M$ increases, the cell search delay for median users first decreases due to the improved SNR and cell search success probability, until the median users could succeed cell search in the first cycle that they initiates IA. Then the cell search delay will increase as $M$ is further increased, which is because the beam sweeping overhead becomes more dominant. The optimal BS antenna number $M$ is $12$ (or $30^\circ$ beamwidth) in Fig.~\ref{50thDelayFigNoiseLmted}, which corresponds to a cell search delay of $0.31$ ms. 

\begin{figure}[h]
	\centering
	\includegraphics[width=0.55\linewidth]{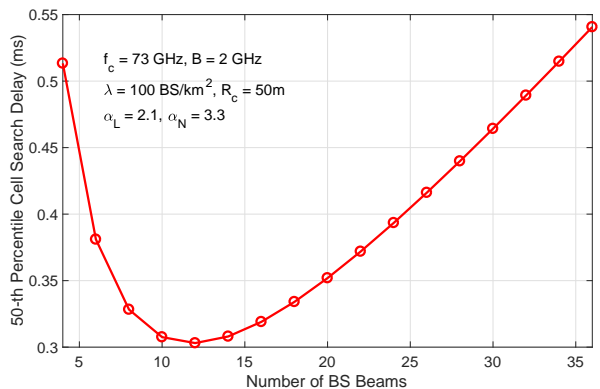}
	\caption{$50^{\text{th}}$ percentile cell search delay for noise limited network.}
	\label{50thDelayFigNoiseLmted}
\end{figure}

\subsection{Cell Search Delay Distribution in Interference Limited Networks}
Similar to the noise limited scenario, we have evaluated the CCDF of cell search delay for the interference limited scenario in Fig.~\ref{CSDelayDistFigIntLmtd} by generating $10^6$ realizations of $R_0$ and computing the corresponding $D_{cs}(R_0,M,\lambda)$ through Corollary~\ref{CSDelayDistIntLmtedCoro}. 

Fig.~\ref{CSDelayDistFigIntLmtd} shows that the tail distribution function of $D_{cs}(R_0,M,\lambda)$ decreases almost linearly under the log-log scale, which means the distribution of $D_{cs}(R_0,M,\lambda)$ is also heavy-tailed under the interference limited scenario. However, in contrast to the noise limited scenario where the overall mean cell search delay is always infinite, the phase transition for mean cell search delay of the interference limited scenario can be observed from Fig.~\ref{CSDelayDistFigIntLmtd}. Specifically, when the cell search is performed omni-directionally (i.e., $M=1$), Fig.~\ref{CSDelayDistFigIntLmtd} shows that the decay rate of the tail satisfies $\lim_{t \rightarrow \infty} \frac{-\log \mathbb{P}(D_{cs}(R_0,M,\lambda) \geq t)}{\log t} < 1$, which indicates an infinite mean cell search delay. As $M$ increases to $4,8, 12$, Fig.~\ref{CSDelayDistFigIntLmtd} shows that $\lim_{t \rightarrow \infty} \frac{-\log \mathbb{P}(D_{cs}(R_0,M,\lambda) \geq t)}{\log t} > 1$, which leads to a finite mean cell search delay. This observation is consistent with Theorem~\ref{PhaseTranIntLmtedRmk}, which shows that for the considered interference limited scenario with path loss exponent $\alpha = 2.5$ and SINR detection threshold $\Gamma_{cs} = -4$ dB, the mean cell search delay is infinite when $M = 1$, and finite as long as $M > 1.59$. 

It can also be observed from Fig.~\ref{CSDelayDistFigIntLmtd} that BS beam-sweeping can significantly reduce the cell search delay for both the median users and edge users in the interference limited networks. For example, when the number of BS antennas/beams $M$ is $1$, $4$, $8$, and $12$, the corresponding cell search delay for the $50^{\text{th}}$ percentile user is $200$ ms, $8.98$ ms, $1.18$ ms, and $0.9123$ ms respectively, while the corresponding cell search delay for the $10^{\text{th}}$ percentile user is $3720$ ms, $53.84$ ms, $5.14$ ms and $1.35$ ms respectively. The main reason for such a performance gain in the interference-limited network is that as $M$ increases, beam-sweeping creates more angular separations from the nearby BSs to the user, so that the number of IA cycles to succeed in cell search can be effectively reduced, especially for edge users.

\begin{figure}[h]
	\centering
	\includegraphics[width=0.55\linewidth]{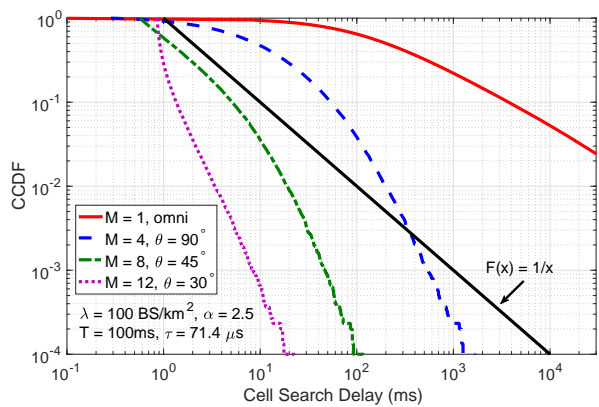}
	\caption{Cell search delay distribution for interference limited network.}
	\label{CSDelayDistFigIntLmtd}
\end{figure}

\section{Conclusions}
This paper has proposed a mathematical framework to analyze the directional cell search delay for fixed cellular networks, where the BS and user locations are static. 
Conditioned on the BS locations, we have first derived the conditional expected cell search delay under the Palm distribution of the user process. By utilizing a Taylor series expansion, we have further derived the exact expression for the overall mean cell search delay in a Poisson cellular network with Rayleigh fading channels. Based on this expression, the expected cell search delay in noise-limited network was proved to be infinite when the NLOS path loss exponent is larger than 2. By contrast, a phase transition for the expected cell search delay in the interference-limited network was identified: the delay is finite when the number of BS beams/antennas is greater than a threshold, and infinite otherwise. Finally, by investigating the distribution of the conditional cell search delay given the distance to the nearest BS, the cell search delay for the edge user was shown to be significantly reduced as the number of BS beams/antennas increases, which holds true for both the noise and interference limited networks.

The framework developed in this paper provides a tractable approach to handle the spatial and temporal correlations of
user's SINR process in cellular networks with fixed BS and user locations. Future work will leverage the proposed framework to derive the random access phase performance, the overall expected initial access delay, as well as the downlink throughput performance for such fixed cellular networks. In addition, we will also extend the framework to incorporate user beamforming or power control. 

\section*{Acknowledgments}
This work is supported in part by the National Science Foundation under Grant No. NSF-CCF-1218338 and an award from the Simons Foundation (\#197982), both to the University of Texas at Austin.

\bibliographystyle{ieeetr}
\bibliography{reference}

\end{document}